\documentclass[submission,copyright,creativecommons]{eptcs}
\providecommand{\event}{Linearity-TLLA 2018} 
\usepackage{breakurl}             
\usepackage{underscore}           
\usepackage{amsmath,amsthm}
\usepackage{amssymb}
\usepackage{dashrule}
\usepackage{mathrsfs}
\usepackage{graphicx}
\usepackage{xy}
\usepackage[pdftex,table]{xcolor}
\usepackage{proof}
\usepackage{stmaryrd}
\usepackage[bbgreekl]{mathbbol}
\usepackage[normalem]{ulem}

\newcommand{\commento}[1]{}
\usepackage[textsize=tiny,shadow,color=white,linecolor=orange]{todonotes}

\newcommand*{\myDef}{\mathrel{\vcenter{\baselineskip0.5ex \lineskiplimit0pt
                     \hbox{\scriptsize.}\hbox{\scriptsize.}}}%
                     \mathrel{\vcenter{\baselineskip0.5ex \lineskiplimit0pt
                     \hbox{\scriptsize.}\hbox{\scriptsize.}}}%
                      =
                     }
\newcommand{\powerset}{\raisebox{.15\baselineskip}{\Large\ensuremath{\wp}}}
  \DeclareRobustCommand*\pct{\scalebox{.9}{\%}}
\newcommand{\PRB}[1]{{#1}\!}

\newcommand{\ket}[1]{|#1\rangle}

\DeclareFontEncoding{LS1}{}{}
\DeclareFontSubstitution{LS1}{stix}{m}{n}
\DeclareSymbolFont{symbols2}{LS1}{stixfrak}{m}{n}
\newcommand{\typecolon}{\mathrel{\!\!:\!:\!\!}}
\newcommand{\wires}{\mathop{\mathsf{wr}}}

\newcommand\dom{\textrm{dom}}
\newcommand\codom{\textrm{ran}}

\newcommand{\CIRC}{\textsf{c\i rc}}

\newcommand{\Nat}{\textsf{Nat}}

\newcommand{\qCom}{\cmd}
\newcommand{\cmd}{\textsf{cmd}}

\newcommand{\qnew}{\mathrel{\textsf{qnew}}}
\newcommand{\cnew}{\mathrel{\textsf{cnew}}}

\newcommand{\QNew}[3]{\qnew^{#2} {#1}\qIn {#3}}
\newcommand{\CNew}[3]{\cnew^{#2} {#1} \qIn {#3}}
\newcommand{\qVar}{\textsf{qVar}}
\newcommand{\cVar}{\textsf{cVar}}
\newcommand{\while}{\mathop{\tt while}}
\newcommand{\wdo}{\mathrel{\tt do}}
\newcommand{\qIn}{\mathrel{\tt in}}

\newcommand{\vettore}[1]{\overline{\mathbb{#1}}}

\newcommand{\revCirc}{\mathop{\mathtt{reverse}}}
\newcommand{\cRead}{\mathop{\textsf{read}}}
\newcommand{\qMeas}{\mathop{\textsf{meas}}}
\newcommand{\partMeas}{\mathop{\textrm{pMeas}}}

\newcommand{\cEval}{\textrm{cEval}}

\newcommand{\qpcf}{\text{\textsf{qPCF}}}
\newcommand{\QuIA}{\textsf{IQu}}
\newcommand{\qdcc}{\emph{quantum data \& classical control~}}

\def\TTY{{\tt Y}}

\newcommand{\condinc}[2]{\ifthenelse{\equal{\typeof}{0}}{#1}{#2}}

\newcommand{\restr}[1]{{\mathop{\upharpoonright_{\! #1}}}}

\def\evaluates{\ev}

\def\fleche{\rightarrow}

\newcommand{\ev}{\mathrel{\pmb\Downarrow}}

\newcommand{\s}{{\bf s}}
\newcommand{\p}{{\bf p}}

\newcommand{\ifz}{\text{\tt if}}
\newcommand{\lif}[3]{\ifz\,\, {\tt #1 }\,\, {\tt #2} \,\, {\tt #3}\,}

\newcommand{\PCF}{\text{\textsf{PCF}}}

\newcommand{\FV}{\mathrm{FV}}

\newcommand\num[1]{{\tt \underline{#1}}}

\newcommand{\n}{\num{n}}
\newcommand\equal[2]{{\tt eq}_{#1}({\tt #2})}

\def\Pred{{\tt pred}}
\def\Succ{{\tt succ}}

\newcommand{\none}{\texttt{skip}}

\newcommand{\rsize}{\mathop{r\texttt{size}}}
\newcommand{\csize}{\mathop{c\texttt{size}}}

\newcommand{\mpar}{\mathsf{M}_{par}}

\newenvironment{myArray}[1][1]{%
  \renewcommand*{\arraystretch}{#1}%
  \array%
}{%
  \endarray
}

\newcommand*\leftPart[2]{ \mathrel{{#1}\hspace{-.8mm}\Lsh_{#2} }}
\newcommand*\rightPart[2]{ \mathrel{{#1}\Rsh_{\hspace{-.8mm}#2} }}

\newtheorem{definition}{Definition}
\newtheorem{example}{Example}
\newtheorem{proposition}{Proposition}

\newtheorem{theorem}{Theorem}

\title{Quantum Programming Made Easy}
\author{Luca Paolini 
\institute{Dipartimento di Informatica\\ Universit\`a di Torino\\ Italy}
\email{luca.paolini@unito.it}
\and
Luca Roversi 
\institute{Dipartimento di Informatica\\ Universit\`a di Torino\\ Italy}
\email{luca.roversi@unito.it}
\and
   Margherita Zorzi
\institute{Dipartimento di Informatica\\ Universit\`a di Verona\\ Italy}
\email{margherita.zorzi@univr.it}
}
\def\titlerunning{Quantum Programming Made Easy}
\def\authorrunning{Paolini, Roversi \& Zorzi}

\begin{document}




 


\maketitle

\begin{abstract}
  We present $ \QuIA $, namely a quantum programming language that extends Reynold's 
  Idealized Algol, the paradigmatic core of Algol-like languages.
  $ \QuIA $ combines imperative programming with high-order features, mediated by a simple type theory.
    $ \QuIA $ mildly merges 
  its quantum features with the classical programming style that we can 
  experiment through Idealized Algol, the aim being to ease a transition towards the quantum  programming world.
  The proposed extension is done along two main directions.
  First, $ \QuIA $ makes the access to quantum 
  co-processors by means of quantum stores. 
  Second, $ \QuIA $  includes some support for the  direct manipulation 
  of quantum circuits, in accordance with recent trends in the development of  quantum programming languages.
 Finally, we show that $ \QuIA $ is quite effective in expressing well-known quantum  algorithms.
\end{abstract}

 
\section{Introduction}\label{sec:intro}

Linearity is an essential ingredient for quantum computing, since quantum data 
have to undergo restrictions such as non-cloning and non-erasing properties.
This is evident from the care that quantum programming language design puts on 
the management of quantum bits, especially in presence of higher-order features.

Selinger's QPL \cite{selinger2004mscs} is a milestone for quantum  programming 
theory. It follows the mainstream  approach ``\qdcc'' based on the 
architecture QRAM~\cite{Knill96}.
In QPL a classical program generates ``directives'' for  an ideal quantum 
device. 
 QPL is the first statically typed programming language that enforces the well-formedness of a program at compile time,
by avoiding run-time checks. The non-duplication of quantum data is enforced by 
the syntax of the language.
The lack of higher-order functions is the main limitation of QPL.

The introduction of QPL opened the way to the design of several quantum 
programming languages.
Some of them came be found  in~\cite{Grattage11entcs,diaz11,pagani14acm,Zorzi16,selinger2006mscs,selinger09chap,LagoZ15}.
Pagani et al.~\cite{pagani14acm} and Zorzi~\cite{Zorzi16} mainly focus
on the computational models behind the language. Other papers, 
like~\cite{selinger2006mscs}, focus on pioneering prototypes of 
effective quantum programming languages. 
The languages in \cite{selinger2006mscs,pagani14acm,Zorzi16,Felty16} deal 
with higher-order functions and 
follow the direction suggested by Selinger in~\cite{selinger2004mscs}:
Linear Logic 
exponential modalities are used to devise typing systems which 
accomodate quantum data management in a classical programming setting.

A very pragmatic proposal is  \texttt{Quipper}~\cite{quipperacm}.
Its distinguishing feature is its management of quantum circuits  as classical data.
It allows to duplicate, erase, operate and meta-operate circuits.
The prototypical example of meta-operations are identified  in~\cite{Knill96}:
reversing quantum circuits, conditioning of quantum operations and converting classical algorithms in reversible ones.
Some formalizations of the core of \texttt{Quipper}, called \texttt{ProtoQuipper} and \texttt{ProtoQuipper-M}, have been defined in   \cite{rios2017eptcs,Ross2015} and are base based on Linear-Logic typing systems.

Moving the focus from the quantum-data perspective to the classical-control one,
we find recent  quantum programming languages as $\qpcf$~\cite{paolini2017lncs} and \texttt{QWire}~\cite{qwire}.
\texttt{QWire} has been conceived as an extension of a classical language
that provides an elegant and manageable language for defining quantum  
circuits; Haskell and Coq have been considered as possible hosts
(in fact, \texttt{QWire} is a ``quantum plugin'' for a host classical language).
It provides a suitable support for quantum circuits (extended with measurement gates) through a boxing interface that rests
on a Linear-Logic based typing system.
If some normalization assumption holds, then the interface keeps the typing rules for elements that live in the quantum 
world apart from the typing system of the hosting language.
In contrast, $\qpcf$~\cite{paolini2017lncs} is a stand-alone quantum programming language. It extends $\PCF$ with a type for quantum circuits
and its operational semantics is supplied by means of a QRAM-compliant device with suitably relaxed de-coherence assumptions.
$ \qpcf $ has two main distinctive type features: it uses dependent types and, it avoids types for quantum states,
by relegating their relevance to quantum co-processor calls.
More precisely, $\qpcf$ prevents the access to intermediate quantum states
by tightening up the interaction process with the black-box quantum device: whole bunches of quantum directives are supplied to the device
that ends their evaluation with a measurement of the whole state.
For the sake of completeness, we recall that dependent type extensions of \texttt{QWire} are considered in \cite{qwire,qwire2}.


\medskip

In this paper, we introduce $\QuIA$ (read ``\emph{Haiku''} as the japanese poetic form) a new 
quantum programming language that extends Idealized Algol by conservatively  inheriting its positive qualities.
Reynolds's Idealized Algol is a paradigmatic language that elegantly combines the fundamental features of procedural languages,
i.e. local stores and  assignments, with a fully-fledged higher-order procedure mechanism which, in its turn, conservatively includes $\PCF$.
Idealized Algol's expressiveness and simple type  theory have attracted the attention of many researchers \cite{book1997ia}.

$\QuIA$ extends  Idealized Algol by introducing two new types.
The first one types quantum circuits, the second one types quantum variables.
Considered that quantum circuits are classical data, their type allows 
to operate on them without any special care, as they were a kind of numerals (roughly, they are special strings).
Instead, manipulating quantum states requires much care, for which we adapt Idealized Algol's original de-reference mechanism to access the 
content of classical variables. In Idealized Algol, a classical variable is a name for a classical value stored in a register.
We cannot duplicate that register, but we can access its content via suitable methods and, if interested to that, we can duplicate that content.
Following this approach, we introduce the type of ``quantum variables'' that prevent to read quantum states, but allow to measure them.
The unique irreversible update of quantum stare allowed by $\QuIA$  js the measurement, while all other transformations are required to be unitary ones.

Essentially, $\QuIA$ can be seen as a higher-order extension of QPL~\cite{selinger2004mscs} where we prevent the quantum state duplication
not by means of a Linear Logic-based typing system, but putting them inside suitable variables, whose content cannot be  duplicated.
However, the duplication can be applied to the name of the variables to safely refer to their content in different spots of a program.
$\QuIA$ manipulates the quantum variables as much as possible as imperative variables. 

According to the recent trend in the development of quantum programming 
languages, we introduce some facilities to duplicate and operate on quantum circuits.
This approach stems from how quantum algorithms are usually described and built, from the 
far-sighted informal ideas proposed in~\cite{Knill96} and, constitutes one of the distinctive features of \texttt{Quipper}, \texttt{QWire} and $\qpcf$.
However, $\QuIA$ is different from \texttt{ProtoQuipper}-like approaches and \texttt{QWire}, 
since the linear management of quantum state does not rest on types based on Linear Logic.

In $\QuIA$ we do not use dependent types for describing quantum circuits.
This is to keep it as simple as possible. 
Therefore, formally, $\QuIA$ does not extend $\qpcf$. However, such an extension is possible, thus we state that $\QuIA$ 
extends $\qpcf$ with classical and quantum stores.

Concluding,  $\QuIA$ is an original 
stand-alone higher-order type-safe quantum programming language.
The philosophy behind its design is to keep programming simple and rooted in the traditional classical programming approach
in accordance with the Idealized Algol design. The proposed extension
should help the transition from the classical programming to the quantum one.
Moreover the language smoothly adapts to the common descriptions of quantum algorithms coherently with what Knill advocates in~\cite{Knill96}.\\
\noindent
\textbf{Outline.} Section~\ref{sec:QIA} introduces the row syntax of $\QuIA$ 
and its typing system with some basic properties.
Section~\ref{sec:eval} provides the details about its operational evaluation 
and its type safety.
Section~\ref{sec:example} concretely uses $\QuIA$ to implement some 
well-known algorithms. The last section is about conclusions and future work.

 \vspace{-3ex}
\section{$\QuIA$: Idealized QUantum language}
\label{sec:QIA}

$\QuIA$ is a prototypical and minimal typed language that combines quantum  commands and states with higher-order functional features by using registers.
It is an extension Idealized Algol (see \cite{book1997ia,ong2004apal}), namely a $\PCF$ that includes assignments and side-effects.

The grammar of $ \QuIA $  \emph{ground types} is  $\beta ::=  \Nat \mid \cVar \mid \qVar \mid \cmd \mid  \CIRC$.
\begin{itemize}
\item $\Nat$ is the type of numerical expressions which evaluate to natural  numbers;
\item $\cVar$ is the type of imperative variables that store natural numbers;
\item $\qVar$ is the type of quantum registers that  store  quantum states, in accordance with the QRAM model
  (so states are assumed free of decoherence issues);
\item $\CIRC$ is the type of quantum-circuit expressions, i.e. expressions 
evaluating to strings that describe unitary transformations;
\item $\qCom$ is the type of commands, i.e. the type of  operations on variables (producing side-effects).
\end{itemize}
\noindent 
The \emph{types} of $ \QuIA $ are the language that
$\sigma, \tau,\theta ::= \beta \mid \theta\rightarrow \theta$ generates.
\\
The \emph{terms} of $\QuIA$ belong to the row syntax:
$$
\begin{array}{rrl}
{\tt M,N,P,Q} &\myDef&  \texttt{x} \mid \num{n} \mid \Pred\mid \Succ   \mid \ifz \mid  {\tt \lambda x.M} \mid {\tt MN} \mid \TTY_\sigma\\ 
              & \mid &\none \mid {\tt M;N} \mid \while \mathtt{P} \wdo \mathtt{Q} \mid \mathtt{M}:= \mathtt{N} \mid  \cRead \mathtt{M}\mid   \CNew{\mathtt{x}}{\mathtt{N}}{\mathtt{M}} \\
              & \mid & \mathtt{U}^\num{k} \mid  \quad\typecolon \quad 
              \mid\quad\parallel \quad\mid \revCirc \mid \csize \mid \rsize 
              \mid  \mathtt{M}\vartriangleleft {\tt N}\mid\qMeas^\mathtt{N}  
              \mathtt{M} \mid   \QNew{\mathtt{x}}{\mathtt{N}}{\mathtt{M}} 
              \enspace .
\end{array} 
$$
\begin{itemize}
\item  The first line includes in $\QuIA$,   a boolean-free call-by-name $\PCF$,
  namely the sub-language of $\QuIA$ which contains variables, numerals, predecessor, successor, conditional, abstraction, application and recursion;
\item The second line includes in $\QuIA$, the imperative part of Idealized Algol 
\cite{ong2004apal,pitts1997chapter} which contains the commands do-nothing, 
composition, iteration, assignment, store reader ($\cRead$) and store binder 
($\cnew$);
\item  The third  line includes in $\QuIA$, the syntax of quantum circuits,
  basic operations on both quantum circuits and quantum registers.\\
  \textbf{Quantum Circuits.} We expect that evaluating a circuit expression yields 
  (evaluated) circuits which are strings generated by the grammar $\mathtt{C} \myDef 
  \mathtt{U}^\num{k} \mid \mathtt{C} \typecolon \mathtt{C} \mid   \mathtt{C}\parallel  \mathtt{C} $  and which include, 
  quantum gates (${\tt U}^\num{k}$ denotes a gate in the set $\mathcal{U} (k)$ of gates operating on $k$ wires), sequential and parallel compositions.
  A non null arity, i.e. the number of its inputs and of its output, labels 
  every gate. For every gate symbol, the quantum co-processor~\cite{NieCh10}
  implements a unitary transformation to interpret it.
  W.l.o.g., we assume that a universal base of quantum transformations is  available.
 The design of circuit's syntax is aimed to be basic and to predispose the use of dependent types in $\QuIA$.\\
 \textbf{Quantum Circuit Operations.} We limit our interest to prototypical quantum circuits operations.
    $\revCirc$ is the standard quantum meta-operation for reversing (see~\cite{Knill96,qwire}).
    $\csize$ returns the arity of a circuit.\\
  \textbf{Quantum Operations.} $\rsize$ returns the arity of a quantum register.
    $\mathtt{x}\vartriangleleft {\tt N}$ evaluates the application of the circuit ${\tt N}$ to the quantum state stored in $\mathtt{x}$,
    then it stores the resulting state in $\mathtt{x}$.
    Roughly, $\qMeas^\mathtt{N}  \mathtt{x}$ measures ${\tt N}$ qubits of a quantum  state which $\mathtt{x}$ stores (and, update such state, in accordance with the quantum measurement rules).
\end{itemize}

\subsection{Typing system}

\begin{table}[t]
\scalebox{.93}{$\begin{myArray}[1.8]{c}
     \infer[\scriptstyle(tv)]{   B \cup  \{\mathtt{x}:\sigma\} \vdash {\tt x}:\sigma }{}  \quad \infer[\scriptstyle(tn)]{B\vdash \num{n}:\Nat}{}   \quad
\infer[\scriptstyle(ts)]{B\vdash {\tt succ}:\Nat\fleche \Nat}{} 
\qquad \infer[\scriptstyle(tp)]{B\vdash {\tt pred}:\Nat\fleche \Nat}{} 
\\[3mm]
\infer[\scriptstyle(tab)]{B\vdash {\tt \lambda x^\sigma.N}:\sigma\fleche\tau }{B\cup  \{{\tt x}:\sigma\}\vdash {\tt N}:\tau }
       \qquad \infer[\scriptstyle(tap)]{B\vdash {\tt PQ}:\tau}{B\vdash {\tt P}:\sigma\fleche\tau\quad B\vdash {\tt Q}:\sigma }
     \qquad \infer[\scriptstyle(tY)]{B\vdash \TTY_\sigma:(\sigma \fleche \sigma) \fleche \sigma} {}\\[3mm]
  \infer[\scriptstyle(ti)]{B\vdash \ifz:\Nat\fleche \beta\fleche\beta\fleche\beta}{} 
 \quad \infer[\scriptstyle(tk)]{B\vdash \none:\qCom}{}  
 \quad\infer[\scriptstyle(tc)]{B\vdash {\tt P;Q}:\beta}{B\vdash {\tt P}:\qCom \quad B\vdash {\tt Q}:\beta} 
 \\[3mm]
  \infer[\scriptstyle(tw)]{B\vdash  \while \mathtt{P} \wdo \mathtt{Q}:\qCom}{B\vdash {\tt P}:\Nat \quad B\vdash {\tt Q}:\qCom}
  \quad \infer[\scriptstyle(tA)]{B \vdash \mathtt{M} := {\tt N}:\qCom }{B \vdash {\tt M}:\cVar \quad B \vdash {\tt N}:\Nat}
  \quad \infer[\scriptstyle(tR)]{B\vdash \cRead \mathtt{M}  :\Nat}{B\vdash {\tt M}:\cVar }
\\[3mm]
\infer[\scriptstyle(tcnw)]{B\vdash  \CNew{\mathtt{x}}{\mathtt{N}}{\mathtt{M}} :\beta }{
  B\vdash \mathtt{N}:\Nat & B\cup  \{{\tt x}:\cVar\}\vdash {\tt M}:\beta} 
\qquad  \infer[\scriptstyle(tc1)]{B\vdash {\tt U}^\num{k}:\CIRC}{ {\tt U}\in\mathcal{U} (k)} 
\\[1mm]
     \infer[\scriptstyle(tc2)]{B\vdash \typecolon:\CIRC\fleche\CIRC\fleche\CIRC }{  }
     \qquad \infer[\scriptstyle(tc3)]{B\vdash \parallel :\CIRC\fleche\CIRC\fleche\CIRC }{  }
     \qquad  \infer[\scriptstyle(tmc)]{B\vdash \revCirc:\CIRC\fleche\CIRC }{ }  
 \\[1.8mm]
 \infer[\scriptstyle(tsc)]{B\vdash \csize (\mathtt{M}) :\Nat }{ B\vdash \mathtt{M}:\CIRC }
 \qquad \infer[\scriptstyle(tsr)]{B\vdash \rsize (\mathtt{M}) :\Nat }{ B\vdash \mathtt{M}:\qVar }
 \qquad  \infer[\scriptstyle(tC)]{B \vdash \mathtt{M} \vartriangleleft{\tt N}:\qCom }{B \vdash {\tt M}:\qVar \quad B \vdash {\tt N}:\CIRC}
\\[3mm]
  \infer[\scriptstyle(tM)]{B\vdash \qMeas^{\tt N} {\tt M} :\Nat}{B\vdash {\tt M}:\qVar \quad B\vdash {\tt N}:\Nat }
\qquad \infer[\scriptstyle(tqnw)]{B\vdash  \QNew{\mathtt{x}}{\mathtt{N}}{\mathtt{M}} :\beta }{
  B\vdash \mathtt{N}:\Nat & B\cup  \{{\tt x}:\qVar\}\vdash {\tt M}:\beta }  
\end{myArray}$}
\rule{\textwidth}{1pt}
\caption{Typing Rules.}
\label{TypingRules}
\end{table}

A \emph{base} is a finite list $\mathtt{x}_1:\sigma_1,\ldots,\mathtt{x}_n:\sigma_n$ that we manage as a set
such that $ \mathtt{x}_i \neq \mathtt{x}_j $ for every $ i\neq j $.  
If $B=\mathtt{x}_1:\sigma_1,\ldots,\mathtt{x}_n:\sigma_n$, 
then $\dom(B)=\{\mathtt{x}_1,\ldots,\mathtt{x}_n\}$ and 
$\codom(B)=\{\sigma_1,\ldots,\sigma_n\}$.
The extension of $ B $ with ${\tt x}:\sigma$ is denoted $B\cup\{{\tt x}:\sigma\}$ where,
w.l.o.g., we assume such that $\mathtt{x}\not\in \dom(B)$.

\begin{definition}\label{defTypeRules}
  A term  is \emph{well-typed} whenever it is the conclusion 
    of a finite derivation built with the rules in Table~\ref{TypingRules}.
\end{definition}

The rules $(tv)$, $(tn)$, $(ts)$, $(tp)$,  $(tab)$, $(tap)$, $(tm)$ and $(ti)$ come, quite directly, from $\PCF$.
We just remark that $(ti)$ extends the conditional in order to operate not only 
on numerals, but also on other ground types (circuits, commands, classical and quantum variables).

The rules $(tk)$, $(tc)$, $(tw)$, $(tA)$, $(tR)$ and $(tcnw)$ are the typical 
imperative extensions that Idealized Algol contains.
The rules $(tk)$ and $(tw)$ types the standard imperative commands exactly as in Idealized Algol.
The rule $(tc)$ serves to consistently \emph{concatenate} commands and to 
\texttt{attach} commands to expressions typed with other ground types.
It is worth to remark that, this second use allows to include side-effects in  the evaluation of expressions.
The rule $(tA)$ gives a type to the assignment of a classical variable. 
The rule $(tR)$ gives a type to the result of \emph{reading} a classical variable.
As in Idealized Algol, the rule $(tcnw)$ allows to declare local variables.
If $\mathtt{x}:\cVar$, then $\CNew{\mathtt{x}}{\mathtt{N}}{\mathtt{M}} $ makes 
a new instance of a classical register available. It is a binder that binds the variable name 
$\mathtt{x}$ and whose scope is $\mathtt{M}$. 
The initial value of $ \mathtt{x} $ depends on $\mathtt{N}$.

The rules  $(tc1)$, $(tc2)$, $(tc3)$ give a type to basic gates and to compositions of quantum circuits.
The rule $(tmc)$ gives a type to a meta-operations on quantum circuits.
The rules $(tsc)$ and $(tsv)$ give a type to the operations that return arity values.
The rule $(tC)$ gives a type to the application of a  quantum circuit to a quantum state.
The rule $(tM)$ gives a type to the measurement of a quantum state stored in a quantum variable.
In accordance with $(tcnw)$, the rule $(tqnw)$ allows to declare local  variables.
If $\mathtt{x}:\qVar$, then $\QNew{\mathtt{x}}{\mathtt{N}}{\mathtt{M}} $ gives a new instance of quantum store.
The variable name $\mathtt{x}$ is a binder whose scope is $\mathtt{M}$.
Moreover, w.l.o.g., our assumptions are that the initial value of the quantum 
variable always is the classical state zero and that $\mathtt{N}$ provides the 
numbers of qubits of the associated register.


The type  system enjoys some basic properties.
First, if $B\vdash {\tt M}:\tau$, then  $\FV({\tt M}) \subseteq dom(B)$.
Second, if $B\vdash {\tt M}:\tau$, then $B'\vdash {\tt M}:\tau$ where $B'$ is 
the restriction of $B$ to $FV( {\tt M})$;
Third, if $B\vdash {\tt M}:\tau$ and $ dom(B)\cap dom(B')=\emptyset$,
then $B\cup B'\vdash  {\tt M}:\tau$, which is the weakening of the base.
Generation lemmas hold too, however, our focus will mainly be on the dynamics
of the typing system than on its logical properties.

\begin{proposition}[Substitution]\label{substLemma}
Let $B, \mathtt{x}: \sigma \vdash {\tt M}:\tau$.
\\
If  $B' \vdash {\tt N}:\sigma$ and $ dom(B)\cap dom(B')=\emptyset$, then  
$B\cup B' \vdash {\tt M}[{\tt N}/\mathtt{x}]:\tau$.
\end{proposition}
\begin{proof}
 Standard, by reasoning inductively on the given typing derivation.
\end{proof}
\noindent
Standard Lemmas of typed subject expansions can be straightforwardly proved too.


\begin{example}
Let $\mathtt{N}:\Nat$ be a term whose unique free variable is $\mathtt{r}:\qVar$. 
Also, let $\mathtt{Not}$ denote the not-gate (a.k.a. Pauli-X gate) and let 
$\mathtt{Id}$ denote the identity gate, both with arity $1$.
Let $\mathtt{M}$ denote the term $(\mathtt{r} \vartriangleleft ( (\mathtt{Not} 
\parallel \mathtt{Id}) \parallel \mathtt{Not})); \mathtt{N}$.
Then, $\mathtt{M}$ has type $\Nat$ by using $(tc)$ (just after an application of $(tC)$).
Moreover, by means of  $(tqnew)$, we can conclude
$\vdash \QNew{\mathtt{r}}{\num{3}}{\mathtt{M}}:\Nat$.
Anticipating the semantics, this means that the variable $\mathtt{r}$ in the 
sub-term $\mathtt{N}$ is associated to a quantum register of $3$ qubits 
initialized to $\ket{101}$.
\qed
\end{example}

\vspace{-0.7ex}
\vspace{-3ex}

\section{Evaluation Semantics } \label{sec:eval}

\vspace{-2ex}

\begin{table}
\scalebox{.91}{    $\begin{myArray}[1.8]{c}
    \infer[(en)]{\vettore{s},\num{n}\ev_{\PRB{1}} \tt \vettore{s},\num{n}}{} 
    \qquad \infer[(es)]{\vettore{s},\s(\mathtt{M})\ev_{\PRB{\alpha}} \vettore{s}',\num{n+1}}{ \vettore{s},\mathtt{M}\ev_{\PRB{\alpha}} \vettore{s}',\n}
    \qquad \infer[(ep)]{\vettore{s},\p(\mathtt{M})\ev_{\PRB{\alpha}} \vettore{s}',\n}{\vettore{s},\mathtt{M}\ev_{\PRB{\alpha}} \vettore{s}',\num{n+1}}
    \\[1.8mm]
 \infer[(e\beta)]{{\vettore{s},(\lambda \mathtt{x}.\mathtt{M})\mathtt{N}}\mathtt{P}_1\cdots \mathtt{P}_m\ev_{\PRB{\alpha}} \vettore{s}',\mathtt{V}}
    {\vettore{s},\mathtt{M}[\mathtt{N}/\mathtt{x}]\mathtt{P}_1\cdots \mathtt{P}_m\ev_{\PRB{\alpha}}  \vettore{s}',\mathtt{V}} 
    \qquad \infer[(e\TTY)]{\vettore{s}, \TTY \mathtt{M} \mathtt{P}_1\cdots \mathtt{P}_m\ev_{\PRB{\alpha}} \vettore{s}',\mathtt{V}}
    {\vettore{s}, \mathtt{M} (\TTY \mathtt{M}) \mathtt{P}_1\cdots \mathtt{P}_m\ev_{\PRB{\alpha}}  \vettore{s}',\mathtt{V}}
    \\[1.8mm]   
    \infer[(e\mathrm{if}_l)]{\vettore{s},\lif{M}{L}{R}\ev_{\PRB{\alpha\cdot\alpha'}} \vettore{s}'',\mathtt{V} }
    { \vettore{s},\mathtt{M}\ev_{\PRB{\alpha}} \vettore{s}',\num{0} &  \vettore{s}',\mathtt{L}\ev_{\PRB{\alpha'}} \vettore{s}'',\mathtt{V}}
    \qquad \infer[(e\mathrm{if}_r)]{ \vettore{s},\lif{M}{L}{R}\ev_{\PRB{\alpha\cdot\alpha'}} \vettore{s}'',\mathtt{V}}
    { \vettore{s},\mathtt{M} \ev_{\PRB{\alpha}} \vettore{s}',\num{n+1} & \vettore{s}',\mathtt{R}\ev_{\PRB{\alpha'}}  \vettore{s}'',\mathtt{V}}
    \\[-5mm]
    \hdashrule[0.5ex]{0.99\linewidth}{1pt}{1pt}
    \\[-.8mm]
    \infer[(esk)]{\vettore{s},\none\ev_{\PRB{1}} \tt \vettore{s},\none}{} 
   \qquad \infer[(e\textbf{;})]{\vettore{s},\mathtt{M};\mathtt{N} \ev_{\PRB{\alpha\cdot\alpha'}} \vettore{s}'',\mathtt{V}}
   {\vettore{s},\mathtt{M}\ev_{\PRB{\alpha}} \vettore{s}', \none \quad \vettore{s}',\mathtt{N}\ev_{\PRB{\alpha'}} \vettore{s}'',\mathtt{V}  }
   \qquad \infer[(eVar)]{\vettore{s},\mathtt{x} \ev_{\PRB{1}} \tt \vettore{s},\mathtt{x}}{} 
    \\[1.8mm]
    \infer[(ew1)]{\vettore{s},\while\mathtt{M} \wdo \mathtt{N} \ev_{\PRB{\alpha\cdot\alpha'\cdot\alpha''}} \vettore{s}''',\none}
    {\vettore{s},\mathtt{M}\ev_{\PRB{\alpha}} \vettore{s}', \num{n+1} \quad \vettore{s}',\mathtt{N}\ev_{\PRB{\alpha'}} \vettore{s}'',\none 
      \quad \vettore{s}'', \while\mathtt{M} \wdo \mathtt{N}\ev_{\PRB{\alpha''}} \vettore{s}''',\none}
     \\[1.8mm]
    \infer[(ew0)]{\vettore{s},\while\mathtt{M} \wdo \mathtt{N} \ev_{\PRB{\alpha}} \vettore{s}',\none}{\vettore{s},\mathtt{M}\ev_{\PRB{\alpha}} \vettore{s}', 0  }
    \qquad \infer[(ecA)]{\vettore{s},\mathtt{M} := {\tt N} \ev_{\PRB{\alpha\cdot\alpha'}} \vettore{s}''[ \mathtt{x}\leftarrow \num{n}],\none}
    { \vettore{s},\mathtt{N}\ev_{\PRB{\alpha}} \vettore{s}', \num{n} & \vettore{s}',\mathtt{M}\ev_{\PRB{\alpha}'} \vettore{s}'', \mathtt{x}} 
    \\[1.8mm]  
      \infer[ecR]{\vettore{s},\cRead\mathtt{M}\ev_{\PRB{\alpha}} \vettore{s}', \vettore{s}'(\mathtt{x}) }{ \vettore{s},\mathtt{M}\ev_{\PRB{\alpha}} \vettore{s}', \mathtt{x}}
      \qquad  \infer[ecN]{\vettore{s},\CNew{\mathtt{x}}{\mathtt{N}}{\mathtt{M}\ev_{\PRB{\alpha\cdot\alpha'}} \vettore{s}''\restr{\mathtt{x}},\mathtt{V} } }
      {  \vettore{s},\mathtt{N}\ev_{\PRB{\alpha}} \vettore{s}', \num{n}  & \vettore{s}'[\mathtt{x}\leftarrow \num{n}],\mathtt{M}\ev_{\PRB{\alpha'}} \vettore{s}'', \mathtt{V}}
      \\[-5mm]
      \hdashrule[0.5ex]{0.99\linewidth}{1pt}{1pt}  
      \\  
  \infer[(e\mathrm{u}_2)]{  \vettore{s},{\tt M}_0  \typecolon {\tt M}_1  \ev_{\PRB{\alpha\cdot\alpha'}}  \vettore{s}'',{\tt C}_0 \typecolon  {\tt C}_1 }
      { \vettore{s},{\tt M}_0\ev_{\PRB{\alpha}} \vettore{s}',{\tt C}_0   \quad \vettore{s}',{\tt M}_1\ev_{\PRB{\alpha'}}  \vettore{s}'',{\tt C}_1
          \qquad \scalebox{.8}{$\wires({\tt C}_0)$, $\wires({\tt C}_1)$ are defined, and $\wires({\tt C}_0)=\wires({\tt C}_1)$} }
        \\[3mm]
      \infer[(e\mathrm{u}_1)]{  \vettore{s}, {\tt U}^\num{k} \ev_{\PRB{1}} \vettore{s}, {\tt U}^\num{k}}{ }
    \quad  
        \infer[(e\mathrm{u}_3)]{\vettore{s},{\tt M}_0\parallel {\tt M}_1\, \ev_{\PRB{\alpha\cdot\alpha'}}    \vettore{s}'',{\tt C}_0 \parallel {\tt C}_1 }
        { \vettore{s},{\tt M}_0\ev_{\PRB{\alpha}} \vettore{s}',{\tt C}_0 \quad \vettore{s}',{\tt M}_1\ev_{\PRB{\alpha'}}  \vettore{s}'',{\tt C}_1  }
\quad  
  \infer[(e\mathrm{r_1})]{\vettore{s},\revCirc {\tt M} \ev_{\PRB{\alpha}} \vettore{s}',\mathtt{U}'}{ \vettore{s},{\tt M}\ev_{\PRB{\alpha}} \vettore{s}',{\tt U} & (\ddagger{\tt U})=\mathtt{U}' }
  \\[1.8mm]
    \infer[(e\mathrm{r_2})]{\vettore{s},\revCirc {\tt M} \ev_{\PRB{\alpha\cdot\alpha'\cdot\alpha''}} \vettore{s}''',{\tt C}'_1  \typecolon {\tt C}'_0 }
    { \vettore{s},{\tt M}\ev_{\PRB{\alpha}} \vettore{s}',{\tt C}_0 \typecolon {\tt C}_1 & \vettore{s}',\revCirc{\tt C}_0\ev_{\PRB{\alpha'}} \vettore{s}'',{\tt C}'_0 
      & \vettore{s}'',\revCirc{\tt C}_1\ev_{\PRB{\alpha''}} \vettore{s}''',{\tt C}'_1 }
    \\[1.8mm]
    \infer[\!\!(e\mathrm{r_3})]{\vettore{s},\revCirc {\tt M} \ev_{\PRB{\alpha\cdot\alpha'\cdot\alpha''}} \vettore{s}''',{\tt C}'_0  \parallel  {\tt C}'_1}
    { \vettore{s},{\tt M}\ev_{\PRB{\alpha}} \vettore{s}',{\tt C}_0 \parallel {\tt C}_1 & \vettore{s}',\revCirc{\tt C}_0\ev_{\PRB{\alpha'}} \vettore{s}'',{\tt C}'_0 
      & \vettore{s}'',\revCirc{\tt C}_1\ev_{\PRB{\alpha''}} \vettore{s}''',{\tt C}'_1 }
      \\[1.8mm]      
      \infer[(e\mathrm{csz})]{ \vettore{s}, \csize \mathtt{M} \ev_{\PRB{\alpha}} \vettore{s}, \wires(\mathtt{C}) }{ \vettore{s},\mathtt{M}\ev_{\PRB{\alpha}} \vettore{s}', \mathtt{C}  }
      \qquad
       \infer[(e\mathrm{rsz})]{ \vettore{s}, \rsize \mathtt{x} \ev_{\PRB{1}} \vettore{s},\vettore{s}_\sharp(\mathtt{x}) }{ }
    \\[1.8mm]
    \infer[(eqA_0)]{\vettore{s},\mathtt{M} \vartriangleleft {\tt N} \ev_{\PRB{\alpha\cdot\alpha'}} \vettore{s}''[ \mathtt{x}\xleftarrow{{\vettore{s}''}_\sharp(\mathtt{x})}\cEval^{\num{n}}\big( {\tt C}\big)\big( \vettore{s}'(\mathtt{x})\big)],\none}
    { \vettore{s},\mathtt{N}\ev_{\PRB{\alpha}} \vettore{s}', {\tt C}   &\vettore{s}',\mathtt{M}\ev_{\PRB{\alpha'}} \vettore{s}'', \mathtt{x} &  \wires({\tt C}) ={\vettore{s}}_\sharp(\mathtt{x})=\num{n} } 
    \quad
    \infer[(eqA_1)]{\vettore{s},\mathtt{M} \vartriangleleft {\tt N} \ev_{\PRB{\alpha\cdot\alpha'}} \vettore{s}'',\none}
    { \vettore{s},\mathtt{N}\ev_{\PRB{\alpha}} \vettore{s}', {\tt C}   &\vettore{s}',\mathtt{M}\ev_{\PRB{\alpha'}} \vettore{s}'', \mathtt{x} &  \wires({\tt C}) \neq{\vettore{s}}_\sharp(\mathtt{x})}
   \\[1.8mm]
 \infer[(eqM)]{\vettore{s},\qMeas^\mathtt{N} \mathtt{M} \ev_{\PRB{\alpha\cdot\alpha'\cdot\alpha''}} \vettore{s}''[ \mathtt{x}\leftarrow \ket{\phi} ] , \num{m} }
{
  \begin{array}{ll}
    \vettore{s},\mathtt{N}\ev_{\PRB{\alpha}} \vettore{s}', \num{k}   & {\vettore{s}}_\sharp(\mathtt{x})=\num{n} \\[-3.8mm]
    \vettore{s}',\mathtt{M}\ev_{\PRB{\alpha'}} \vettore{s}'',\mathtt{x} & (m,\ket{\phi}, \alpha'')\in\partMeas^{\num{n}}(\vettore{s}''(\mathtt{r}), k)
  \end{array}\\[-1.8mm]
}
    \quad
    \infer[(eqN)]{\vettore{s},\QNew{\mathtt{x}}{\mathtt{N}}{\mathtt{M}} \ev_{\PRB{\alpha\cdot\alpha'}} \vettore{s}''\restr{\mathtt{x}},\mathtt{V}}
                    {\vettore{s},\mathtt{N}\ev_{\PRB{\alpha}} \vettore{s}', \num{n} & \vettore{s}'\cup\{\mathtt{x}\xleftarrow{\num{n}}0\},\mathtt{M}\ev_{\PRB{\alpha'}} \vettore{s}'',\mathtt{V} }
                  \end{myArray}$
                  }
    \rule{\textwidth}{1pt}
    \caption{Operational Semantics.}
    \label{evaluationRules}
\end{table}

The evaluation of $\QuIA$ focuses on programs, i.e. closed terms whose type is ground.
Like Idealized Algol, the operational semantics of $ \QuIA $ uses an auxiliary function to record
 values associated to variables whose type is $\cVar$ or $\qVar$.
In addition, variables of type $\qVar$ must record the number of available qubits.

\begin{definition}[Stores]
A store $ \vettore{s}$ is the disjoint union of two partial functions, 
both defined on a finite domain. The first one has $\cVar$ as domain and 
numerals as co-domain. The second function has $\qVar$ as domain and pairs 
(quantum states, numerals) as co-domain, where the second component counts the 
qubits of the first component.
If $\mathtt{x}:\cVar$, then $ \vettore{s}(\mathtt{x}) $ denotes the numeral stored in $\mathtt{x}$.
If $\mathtt{x}:\qVar$, then $ \vettore{s}(\mathtt{x}) $ denotes 
the state stored in $\mathtt{x}$ and $ \vettore{s}_\sharp(\mathtt{x}) $ the number of qubits of the state $ \vettore{s}(\mathtt{x}) $.
The domain of definition of $\vettore{s}$ is denoted $\dom(\vettore{s})$, i.e. it is a set of variable names. 

We denote $\vettore{s}\restr{\mathtt{x}}$ the store that 
behaves like $\vettore{s}$ everywhere, except  on ${\tt x}$ where it is 
undefined.
Let $\vettore{s}$ be a store and let $\mathtt{x}:\cVar$;
then $\vettore{s}[{\tt x}\leftarrow\num{k}]$ denotes a store that behaves 
like $\vettore{s}$ everywhere except than on ${\tt x}$ to which it associates
$\num{k}$. Let $\vettore{s}$ be a store and let $\mathtt{x}:\qVar$;
then $\vettore{s}[{\tt x}\xleftarrow{\num{k}}\ket{\phi}]$ denotes a store that 
behaves like $\vettore{s}$ everywhere except than on ${\tt x}$ on which we have 
that $\vettore{s}(\mathtt{x}) = \ket{\phi}$ and 
$ \vettore{s}_\sharp(\mathtt{x}) = \num{k}$.
\qed
\end{definition}

We conventionally assume that $\mathtt{C}$ ranges over the strings that describe 
evaluated circuit expressions, i.e.~parallel and series composition of names 
for gates (cf. Theorem \ref{th:progress}).
Moreover, $\mathtt{V}$ ranges over numerals, strings that describe circuits,
names of registers and the command $\none$.

\begin{definition}[Operational Evaluation]
Let 
$ \mathtt{x}_1:\cVar \ldots, \mathtt{x}_n:\cVar, \mathtt{z}_1:\qVar \ldots,  \mathtt{z}_m:\qVar \vdash \mathtt{M}:\beta$ where $n,m \geq 0$ and 
$\beta\in\{\Nat,\CIRC,\cmd,\cVar,qVar\}$.
If $\vettore{s}$ is a store such that $\{\mathtt{x}_1, \ldots ,\mathtt{x}_n\}\subseteq  \cVar\cap \dom(\vettore{s})$ and
$\{\mathtt{z}_1, \ldots ,\mathtt{x}_n\}\subseteq \qVar\cap\dom(\vettore{s})$, 
then the evaluation semantics of $\QuIA$ proves formal statements
$\vettore{s},\mathtt{M} \evaluates_{\PRB{\alpha}}  \vettore{s}',\mathtt{V}$ 
that we obtain as conclusion of a (finite) derivation  $\mathcal{D}$ built with the rules in Table~\ref{evaluationRules}.
As expected, $\alpha\in(0,1]$ is the probability to  obtain $\mathcal{D}$.
\qed
\end{definition}

The rules $(en)$, $(es)$, $(ep)$, $(e\beta)$, $(e\TTY)$, $(e\mathrm{if}_l)$, 
$(e\mathrm{if}_r)$ at the top of Table~\ref{evaluationRules} are standard, that are used in the evaluation of $\PCF$ (e.g., see~\cite{paolini2006iandc}).
They are enriched by a store that can be 
eventually used in the evaluation of their sub-terms (involving side-effects).

The rules $(esk)$, $(e\textbf{;})$, $(eVar)$, $(ew1)$, $(ew0)$, $(ecA)$, $(ecR)$, $(ecN)$ in the middle of Table~\ref{evaluationRules} formalize the 
standard evaluation of first order references and of usual imperative instructions that we find in Idealized Algol. 
The evaluation $(ecA)$ begins by evaluating the expression whose result must be  stored.
The right-hand expression is expected to yield an classical variable.
The rule $(ecA)$ is the only one that changes the content of classical variable.
W.l.o.g., we assume that $\mathtt{x}\not\in \dom(\vettore{s})$ holds for the rule $(ecN)$. Otherwise, 
we had to replace $\vettore{s}''\restr{\mathtt{x}},\mathtt{V}$ with 
$\vettore{s}''[ \mathtt{x}\leftarrow\vettore{s}(\mathtt{x})],\mathtt{V}$
in its conclusion.
Finally, we observe that only the rule that changes the classical variables in the domain of the store is $(ecN)$.

\begin{example}
  \begin{itemize}
  \item Let $\mathtt{P}$ be the well-typed term  $B\cup\{\mathtt{x}:\cVar\}\vdash \ifz \, (\cRead\,\mathtt{x})\,
    \mathtt{M}_0 \,\mathtt{M}_1:\beta$ such that  $B\cup\{\mathtt{x}:\cVar\}\vdash\mathtt{M}_i:\beta$ ($i=0,1$).
    Let us evaluate $\mathtt{P}$ by using the store $\vettore{s}$: if $\vettore{s}(\mathtt{x})=0$ then we have to evaluate
    $\mathtt{M}_0$ (and we ignore $\mathtt{M}_1$), otherwise we  evaluate $\mathtt{M}_1$.
\item In the above term, let $\mathtt{M}_0=\mathtt{z}_0$, $\mathtt{M}_1=\mathtt{z}_1$, and  $\beta=\cVar$.
Starting with the store $\vettore{s}[\mathtt{x}\leftarrow 0]$,
there is a unique derivation describing the evaluation of $(\ifz \,\cRead(\mathtt{x})\, \mathtt{z}_0 \,\mathtt{z}_1): = 5$.
  This derivation concludes 
    $\vettore{s}[\mathtt{x}\leftarrow 0], (\ifz \,\cRead(\mathtt{x})\, \mathtt{z}_0 \,\mathtt{z}_1) : = 5\ev_{\PRB{1}} \vettore{s}[\mathtt{x}\leftarrow 0,\mathtt{z}_0\leftarrow 5],\none$.\qed
\end{itemize}

\end{example}

\paragraph{Extending Idealized Algol}

The remaining rules formalizes our original extension of $\QuIA$. Rules about circuit evaluations are inspired by similar operator of $\qpcf$ (see \cite{paolini2017lncs}).

It is worth to remark that some of these rules rest on some auxiliary definitions (cf. Definitions~\ref{def:wires},  \ref{def:unitSemantic} and~\ref{def:pmeas}).
Definition~\ref{def:wires} has been formalized only for simplicity reasons:
we isolated the function $ \wires $ that counts the number of wires of an evaluated quantum circuit (since it is used in many rules of  Table~\ref{evaluationRules}). 
Definitions~\ref{def:unitSemantic} and~\ref{def:pmeas} formalize the quantum co-processor as an external black-box.

The rules  $(e\mathrm{u}_1)$, $(e\mathrm{u}_2)$, $(e\mathrm{u}_3)$ are used to bring the evaluation of a circuit on subexpressions,
in order to reach (possible) side-effects (assignments and measurements) embedded in it.
The rule $(e\mathrm{u}_2)$ exploits the function $ \wires $ here below.
The rule $(e\mathrm{u}_2)$ applies only when the  $ \wires $ yields the same value on the two circuit components.
On the other hand, $\parallel$ is not subject to arity restriction (cf. rule $(e\mathrm{u}_3)$).

\begin{definition}\label{def:wires}
  $\wires$  is a partial function  from circuits to numerals. It is defined by cases as follows:
    \begin{itemize}
    \item $\wires(\mathtt{U}^\num{k})= \num{k}$;
    \item if $ \wires(\mathtt{M}_0 )$, $\wires(\mathtt{M}_1 )$ are defined and   $ \wires(\mathtt{M}_0 )= \wires(\mathtt{M}_1 )$ then
  $\wires(\mathtt{M}_0 \typecolon \mathtt{M}_1)= \wires(\mathtt{M}_0 )$;
  \item  if $ \wires(\mathtt{M}_0 )$, $\wires(\mathtt{M}_1 )$ are defined then $\;\wires(\mathtt{M}_0 \parallel \mathtt{M}_1)= \wires(\mathtt{M}_0 )+\wires(\mathtt{M}_1)$.
  \end{itemize}
The function is undefined in all other cases.
\qed
\end{definition}

It is easy to check that evaluated circuits (viz. circuits resulting from the evaluation of circuit expression)  are strings of the grammar $\mathtt{C} \myDef \mathtt{U}^\num{k} \mid \mathtt{C} \typecolon \mathtt{C} \mid \mathtt{C}\parallel  \mathtt{C} $  for which $\wires({\tt C})$ is defined.
Note that: (i) if $\mathtt{C}$ is $\mathtt{C}_0 \typecolon \mathtt{C}_1$, then $\wires({\tt C}_0)=\wires({\tt C}_1)$, cf. rule  $(e\mathrm{u}_2)$;
and, (ii) if $\tt M$ is a circuit expression such that $\wires(\mathtt{M})$ is undefined then its evaluation diverges.

We remark that the evaluation of circuits evolves rightward once the possible side-effects in sub-terms are concerned.
We mean that $(e\mathrm{u}_2)$ and $(e\mathrm{u}_3)$ update the store  $\vettore{s}'$ by first evaluating the side-effects in expression of the 
left-hand circuit and, then, by evaluating the side-effects in the expression of the right-hand circuit.

Recent quantum programming languages  \cite{green2013acm,paolini2017lncs,qwire,rios2017eptcs,Ross2015}
include the possibility to manipulate quantum circuits and, in particular, of reversing circuits as originally advocated by \cite{Knill96}.
$\revCirc$ is expected to produce the adjoint circuit of its input: it is implemented by rewiring gates in reverse order
(by means of rules $(e\mathrm{r_1})$, $(e\mathrm{r_2})$ and $(e\mathrm{r_3})$) and, then, by replacing each gate by its adjoint.
Its definition rests on the choice of a total endo-function (mapping each gate of arity $k$ to a gate of arity $k$) that we denote with the symbol $\ddagger$. 
As usual, we assume that $\QuIA$ is endowed with a universal set of gates and that $\ddagger$ gives back an adjoint gate for each gate.
If those latter assumptions do not hold, then $\ddagger$ can be chosen as the identity (so that, $\revCirc$ does not reverse anymore the corresponding quantum transformations,
but just rewires in reverse order).

The rule $(e\mathrm{csz})$ yields the arity of the  quantum (evaluated)  circuit.
The rule $(e\mathrm{rsz})$ yields the number of qubits that the quantum register it involves stores.

\subsubsection*{Quantum State Updates and Quantum Measurements}
We refer to \cite{NieCh10} as a standard and comprehensive reference about  quantum computation.
Here we just recall what we need to introduce the interaction with quantum co-processors.
The information of $n\in\mathbb{N}$ qubits is usually formalized by means of a normalized vector in $\mathcal{H}^n$
which is a Hilbert space of  dimension $N=2^n$ with orthonormal basis:\\[-5.8mm]
$$
\vspace{-2.1mm}\overbrace{\ket{\underbrace{0\ldots 0}_n},\ldots,\ket{\underbrace{1\ldots 1}_n}}^N
\enspace .
$$  
\noindent
The binary representation $x^\flat$ of any value $ x $ in the interval $ [0, \ldots, N-1] $ identifies a vector in $\mathcal{H}^n$.
The binary representation is handy to represent every state $\ket{\psi}$ of the Hilbert space 
as a linear combination:
$$\ket{\psi}=c_0 \ket{0^\flat}+ \ldots + c_{N-1}\ket{(N-1)^\flat}$$
with $c_0,\ldots,c_{N-1}\in\mathbb{C}$. 
Quantum transformations and measurement can transform quantum states~\cite{VVZ14,VVZ17,MVZ11}.
We recall that a measurement reduces a quantum state partially, or totally, to a classical state.
More precisely, given a state $\ket{\phi}\in\mathcal{H}^n$, we can measure a 
subset of qubits in $\ket{\phi}$ (i.e. a partial measurement). 
 The result of the measurement is a residual state vector with a given probability (cf. Definition \ref{def:pmeas}).

 Our co-processor is a black-box,
 so we can formalize the interaction of $\QuIA$ with the quantum co-processor in  an abstract way.
 We update the quantum state associated to a variable that can be updated by means of a function (see the next definition)
 that maps a circuit in the corresponding unitary transformation.
 
\begin{definition}\label{def:unitSemantic} 
Let $\text{Circ}^n$ be the set of evaluated circuits with type $\CIRC$, with  arity $\num{n}$ such that $N=2^{n}$.
Let $\mathcal{H}^n$ be a  Hilbert space of  finite dimension $N$. Let $\{\ket{\varphi_i}\}$ be an orthonormal basis on 
$\mathcal{H}^n$ and let $\mathcal{H}^n \rightarrow\mathcal{H}^n$ be the set of unitary operators on $\mathcal{H}^n$.
The map from evaluated circuits to their corresponding unitary transformations 
is $\cEval^n: \text{Circ}^n \rightarrow (\mathcal{H}^n 
\rightarrow\mathcal{H}^n)$, which we define as follows:
\begin{itemize}
\item $\cEval^n(\mathtt{U}^\num{n})\myDef\mathbf{U}$ where $\mathbf{U}: \mathcal{H}^n \rightarrow\mathcal{H}^n$ is the unitary transformation associated to the gate  $\mathtt{U}$;
\item  $\cEval^n({\tt C}_0 \typecolon {\tt C}_1)\myDef\cEval^n({\tt C}_1)\circ \cEval^n({\tt C}_0)$;
  \item  $\cEval^n({\tt C}_0 \parallel {\tt C}_1)\myDef\cEval^{n_0}({\tt  C}_0)\otimes \cEval^{n_1}({\tt C}_1)$ where $\wires({\tt C}_i)=\num{n}_i$  
  and $n=n_0+n_1$. \qed
\end{itemize}
\end{definition}

The rules $(eqA_0)$ and $(eqA_1)$ in Table~\ref{evaluationRules} evaluate two expressions.
If these evaluations converge then, the first one returns a circuit ${\tt C}$  and the second one returns a quantum variable ${\tt x}$ .
If $\wires({\tt C}) ={\vettore{s}}_\sharp(\mathtt{x})$
then the evaluation proceeds by rule $(eqA_0)$ that uses the function of Definition \ref{def:unitSemantic} to update the corresponding quantum state;
otherwise, the rule $(eqA_1)$ forgot the circuit transformations, in order to ensure the type-safety.

 
  $ \QuIA $ allows for \emph{partial measurements} of an arbitrary subset with $k$ qubits of a $ n $-qubits state.
    
\begin{definition}\label{def:pmeas}
  For  all $k,n\in\mathbb{N}$, let $\leftPart{n}{k}= k\pct (n+1)$ and $\rightPart{n}{k} = n-(k\pct (n+1))$, thus $\leftPart{n}{k}+\rightPart{n}{k}=n$ where $\pct$ denotes the  modulo arithmetic operation.
  If $x<2^j$  then, we use $\flat^j(x)$ to denote the binary representations of $x$  deployed on $j$ bits (i.e. binary digits).
  Moreover, let $S(j)= \{ \flat^j(x)  | 0\leq x <2^j\}$   and  $i\cdot j$ to denote the juxtaposition of $i$ and $j$.
  Following~\cite{KLM07},  we formalize
  $\partMeas^{n}: \mathcal{H}^{n}  \times \mathbb{N} \longrightarrow  \powerset\big( \mathbb{N} \times \mathcal{H}^{n}\times\mathbb{R} \big)$ 
   as follows:
    $$\partMeas\!^{n}(\ket{\phi}, k)= \left\{ 
    \big( m^\natural, \ket{\psi_m}, p_m \big) \left| 
    \begin{array}{l}
      \ket{\phi}=\sum_{i\in S(\leftPart{n}{k})}\sum_{j\in S(\rightPart{n}{k})}c_{i\cdot j} \ket{i}\otimes \ket{j} \text{ and,}\\[3mm]
    m\in S(\leftPart{n}{k}) \text{ s.t. }
    \ket{\psi_m}=\sum_{j \in S(\rightPart{n}{k})} \frac{c_{m \cdot j}}{\sqrt{p_m}} \ket{m}\otimes \ket{j}\\[1mm]
    \hspace{3cm}\text{where }  p_m=\sum_{j \in S(\rightPart{n}{k})} |c_{i\cdot j}|^2
    \end{array}
  \right\}\right.$$
\noindent
where $x^{\natural}$ is the natural number encoded in the sequence of bits $x$. 
The first argument of $ \partMeas\!^{n} $ is a quantum state $\ket{\phi}$ of  $\mathcal{H}^{n}$.  
The second argument  is the number of qubits we  want to measure, modulo $n+1$.  
The result of $ \partMeas^n(\ket{\phi},k) $ is a set of triples.
The first component of the triple is a partial measure executed on $\ket{\phi}$:
its value $m\in\mathbb{N}$ is the measurement of its first $k\pct (n+1)$ qubits. 
The second component  is the collapsing state (still in $\mathcal{H}^{n}$) and it is obtained from  $\ket{\phi}$  by collapsing its measured sub-state to $m$.
The third component is the probability of measuring the value $m$. \qed
\end{definition}
We remark that $\partMeas^{n}(\ket{\phi},0)$ measures $0$ qbit and $\partMeas^{n}(\ket{\phi},n)$ measures all $n$ qbits. 
\begin{example}
Let us consider the state 
$\ket{\phi}=\frac{1}{\sqrt{2}}\ket{0010} +\frac{1}{\sqrt{4}}\ket{1011}+\frac{1}{\sqrt{4}}\ket{0010}$
with 4-qubits in it. Let us  measure its two first quantum bits. 
We remark that $\partMeas^{n}(\ket{\phi}, k)=\partMeas^{n}(\ket{\phi}, n+k)$, 
so that $\partMeas^{4}(\ket{\phi}, 2)=\partMeas^{4}(\ket{\phi}, 6)$.
The partial observation of the two first qbits $\partMeas^4(\ket{\phi},2)$ 
has two possible outcomes. The first triple is 
 $(m_0,\sqrt{\frac{2}{3}}\ket{0010}+\sqrt{\frac{1}{3}}\ket{0001},\frac{3}{4})$
where $m_0^{\flat}=00$. The second one is $(m_1, \ket{1011},\frac{1}{4})$ where
$m_1^{\flat}=10$.
\qed \end{example}

The rule $(eqM)$ first evaluates two expressions in order to obtain a numeral $\num{k}$
and the quantum variable $\tt x$ subject of the measurement.
The numeral $\num{k}$ (modulo the number of qubits stored in  $\tt x$) identifies the number of qubits we want to measure.
Then, it measures the quantum state that the store associates to by using the abstract function of Definition \ref{def:pmeas} which describes the expected behavior of 
the co-processor and that works in accordance with quantum computations laws.
It is worth to note that the only non-deterministic rule of $\QuIA$ evaluation is $(eqM)$;
that is, the probability-label in the conclusion of programs that  never perform quantum measurement is $1$.
We here do not focus on any analysis about  the probabilistic behavior of $ \QuIA $.
We simply remark that $\QuIA$ is an extension of the probabilistic Idealized Algol in~\cite{danos2002acm} and of
the quantum language in~\cite{paolini2018arXivQPCF}. The probabilistic operational equivalence notions can be easily adapted to $\QuIA$.

W.l.o.g., we assume that $\mathtt{x}\not\in \dom(\vettore{s})$ in $(eqN)$. 
Otherwise
$\vettore{s}''\restr{\mathtt{x}},\mathtt{V}$  in its conclusion, should  be replaced with $\vettore{s}''[ \mathtt{x}\xleftarrow{\num{n}}\vettore{s}(\mathtt{x})],\mathtt{V}$ in its conclusion.
Both $(ecN)$ and $(eqN)$ are the only rules that 
change the domain of the store. A programmer can ask for a new quantum 
co-processor for manipulating a quantum state by means of $(eqN)$ at run-time. 
We notice that no limit exists on the number of quantum registers that a 
program in $ \QuIA $ manipulates.

\begin{example}[Bell state circuit]\label{ex:bell}
  We show how $ \QuIA $ can encode a circuit description and evaluation.
   
  The Bell states (or EPR states or EPR pairs) are the simplest examples of entanglement of quantum states~\cite{NieCh10}.
  The following circuit  applies a Hadamard gate on the top wire followed by a controlled-not:
      \begin{center}
        \scalebox{0.35}
        {
          \ifx\pdfoutput\undefined 
          \epsfbox{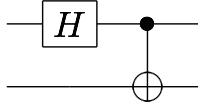} \else 
          \includegraphics{BEll} \fi
        }
      \end{center}
It can be used to generate the Bell states by  feeding it with a bases state 
$\ket{00}$, $\ket{01}$, $\ket{10}$, $\ket{11}$.    
For example, the circuit returns the state $\beta_{00}=\frac{1}{\sqrt{2}}(\ket{00}+\ket{11})$  on  input $\ket{00}$.

Let $H^{\num{1}}:\CIRC$ be the (unary) Hadamard gate, $\texttt{Id}^{\num{1}}:\CIRC$ be identity and  $\texttt{CNOT}^{\num{2}}:\CIRC$ be the controlled-not operator.
Let $\mathsf{Bell}$ be the closed term $ (\texttt{H}^{\num{1}} \parallel\! \texttt{Id}^{\num{1}})\typecolon \texttt{CNOT}^{\num{2}}$
that straightforwardly describes the above circuit. It is easy to check, by typing rules, that
$\vdash  (\texttt{H}^{\num{1}} \parallel \texttt{Id})\typecolon \texttt{CNOT}^{\num{2}}: \CIRC$.

We can simulate an  EPR experiment by using the (closed) term $\QNew {\tt r}{2} {  ( {\tt r}\vartriangleleft  \mathsf{Bell}; \qMeas^\texttt{\num{1}} {\tt r})}$:
it requires that a fresh co-processor (locally, named $r$) is made available for the computation of its body, i.e. ${\tt r}\vartriangleleft  \mathsf{Bell}; \qMeas^\texttt{\num{1}} {\tt r} $.
This latter applies the gates in $\mathsf{Bell}$ to the state stored in ${\tt r}$ and then performs a measurement.\\
Clearly, $\vdash \QNew{\tt r}{2}  {  ( {\tt r}\vartriangleleft  \mathsf{Bell}}; \qMeas\!^\texttt{\num{1}} {\tt r}) :\Nat$
and $\{\big( {\tt r}, \ket{00}\big)\}, {\tt r}\vartriangleleft \mathsf{Bell} \ev_{1}  \{\big({\tt r}, \frac{1}{\sqrt{2}}(\ket{00}+\ket{11})\big)\},\none$.
Moreover, since $\partMeas^2( \frac{1}{\sqrt{2}}(\ket{00}+\ket{11}) , 1)=\{(0,\ket{00},\frac{1}{2}), (1,\ket{11},\frac{1}{2})\}$,
either
$$\{\big({\tt r}, \frac{1}{\sqrt{2}}(\ket{00}+\ket{11}) \big)\}, \qMeas\!^\texttt{\num{1}}{\tt r}  \ev_{\frac{1}{2}}  \{\big({\tt r}, \ket{00}\big)\}, \num{0} \,\,
\mbox{ or } \,\, \{\big({\tt r}, \frac{1}{\sqrt{2}}(\ket{00}+\ket{11}) \big)\}, \qMeas\!^\texttt{\num{1}}{\tt r}  \ev_{\frac{1}{2}}  \{\big({\tt r}, \ket{11}\big)\}, \num{1}\,.$$\qed
 \end{example}

\vspace{-4ex}

\subsection{Type-safety} 

$\QuIA$ enjoys of usual properties of programming languages such as Preservation and  Progress~\cite{pierce2002mit}.
These results follow by adapting, quite straightforwardly, the standard techniques used for $\PCF$ and Idealized Algol.

\begin{theorem}[Preservation]
 Let   $ \mathtt{x}_1:\cVar \ldots, \mathtt{x}_n:\cVar, \mathtt{z}_1:\qVar 
 \ldots, \mathtt{z}_m:\qVar \vdash \mathtt{M}:\beta$ be a term.
 Let   $\vettore{s}$ be a store such that $\{\mathtt{x}_1, \ldots 
 ,\mathtt{x}_n\}\subseteq  \cVar\cap \dom(\vettore{s})$ and
 $\{\mathtt{z}_1, \ldots ,\mathtt{x}_n\}\subseteq \qVar\cap\dom(\vettore{s})$.\\
 If $\vettore{s},\mathtt{M} \evaluates_{\PRB{\alpha}} \vettore{s}',\mathtt{V}$ then
 $ \mathtt{x}_1:\cVar \ldots, \mathtt{x}_n:\cVar, \mathtt{z}_1:\qVar \ldots, \mathtt{z}_m:\qVar \vdash \mathtt{V}:\beta$. 
\end{theorem}
\begin{proof}
The proof is  by induction on the derivation concluding 
$\vettore{s},\mathtt{M} \evaluates_{\PRB{\alpha}} \vettore{s}',\mathtt{V}$.
The proof immediately holds on $(en)$, $(es)$, $(ep)$ while it is true on 
$(e\beta)$, $(e\TTY)$ by arguments related to the inductive hypothesis and the 
subject reduction. The inductive hypothesis straightforwardly applies to 
$(e\mathrm{if}_l)$, $(e\mathrm{if}_r)$, so we can skip to consider the 
imperative part of the language. If the last rule is one among $(esk)$, 
$(eVar)$,  $(ew1)$, $(ew0)$, $(ecA)$, $(ecR)$ the proof is once again 
immediate. If the last rule is  $(e\textbf{;})$, $(ecN)$ it is simple to apply 
the inductive argument. 
The rules $(e\mathrm{u}_1)$, $(e\mathrm{r}_1)$, $(e\mathrm{csz})$, 
$(e\mathrm{rsz})$, $(eqA_0)$, $(eqA_1)$, $(eqM)$ do not pose any specific 
difficulties. Finally, the inductive principle applies also to 
$(e\mathrm{u}_2)$, $(e\mathrm{u}_3)$, $(e\mathrm{r_2})$, $(e\mathrm{r_3})$ and 
$(eqN)$.
\end{proof}

\begin{theorem}[Progress]\label{th:progress}
Let   $ \mathtt{x}_1:\cVar \ldots, \mathtt{x}_n:\cVar, \mathtt{z}_1:\qVar 
\ldots, \mathtt{z}_m:\qVar \vdash \mathtt{M}:\beta$ be a term.
Let   $\vettore{s}$ be a store such that $\{\mathtt{x}_1, \ldots 
,\mathtt{x}_n\}\subseteq  \cVar\cap \dom(\vettore{s})$ and
$\{\mathtt{z}_1, \ldots ,\mathtt{x}_n\}\subseteq \qVar\cap\dom(\vettore{s})$.
\begin{enumerate}
  \item If $\beta=\Nat$ and  $\vettore{s},\mathtt{M} \evaluates_{\PRB{\alpha}} \vettore{s}',\mathtt{V}$ then $\mathtt{V}$ is a numeral.
  \item If $\beta=\CIRC$ and $\vettore{s},\mathtt{M} \evaluates_{\PRB{\alpha}} \vettore{s}',\mathtt{C}$ then $\mathtt{C}$ is a string of the grammar 
    $\;\mathtt{C} \myDef \mathtt{U}^\num{k} \mid \mathtt{C} \typecolon \mathtt{C} \mid \mathtt{C}\parallel  \mathtt{C} $
    such that $\wires({\tt C})$ is defined, and moreover if $\mathtt{C}$ has shape $\mathtt{C}_0 \typecolon \mathtt{C}_1$ then
$\wires({\tt C}_0)=\wires({\tt C}_1)$. 
  \item If $\beta=\cVar$ and $\vettore{s},\mathtt{M} \evaluates_{\PRB{\alpha}} \vettore{s}',\mathtt{V}$ then $\mathtt{V}$  is the name of a classical variable.
  \item  If $\beta=\qVar$ and $\vettore{s},\mathtt{M} \evaluates_{\PRB{\alpha}} \vettore{s}',\mathtt{V}$ then $\mathtt{V}$  is the name of a quantum variable.
    \item If $\beta=\cmd$ and $\vettore{s},\mathtt{M} \evaluates_{\PRB{\alpha}} \vettore{s}',\mathtt{V}$ then $\mathtt{V}$ is $\none$.
  \end{enumerate}
\end{theorem}
\begin{proof}
  All the proofs are by induction on the evaluation, proceeding by cases the 
  subject of the last rule applied.
  \begin{enumerate}
  \item The proof is standard for the $\PCF$ core of $\QuIA$.
    The proof is immediate for $(en)$, $(es)$, $(ep)$, $(ecR)$, 
    $(e\mathrm{csz})$, $(e\mathrm{rsz})$, $(eqM)$.
    The proof for $(e\beta)$, $(e\TTY)$, $(e\mathrm{if}_l)$, 
    $(e\mathrm{if}_r)$, $(e\textbf{;})$, $(ecN)$ and $(eqN)$ follows from the 
    inductive argument. The other cases are not possible because of the typing 
    rules.
  \item 
     The proof is immediate for $(e\mathrm{u}_1)$, $(e\mathrm{r}_1)$.
     The proof for $(e\beta)$, $(e\TTY)$, $(e\mathrm{if}_l)$, 
     $(e\mathrm{if}_r)$,  $(e\textbf{;})$, $(ecN)$,  $(e\mathrm{u}_2)$,  
     $(e\mathrm{u}_3)$, $(e\mathrm{r_2})$, $(e\mathrm{r_3})$, $(eqN)$ follows
     from the inductive argument.
     The other cases are not possible because of the typing rules.
   \item  
     The proof is immediate for $(eVar)$.
     The proof for $(e\beta)$, $(e\TTY)$, $(e\mathrm{if}_l)$, 
     $(e\mathrm{if}_r)$,  $(e\textbf{;})$, $(ecN)$, $(eqN)$  follows
     from the inductive argument.
     The other cases are not possible because the typing rules.
      \item  
     The proof is immediate for $(eVar)$.
     The proof for $(e\beta)$, $(e\TTY)$, $(e\mathrm{if}_l)$, 
     $(e\mathrm{if}_r)$,  $(e\textbf{;})$, $(ecN)$, $(eqN)$  follows
     from the inductive argument.
     The other cases are not possible because of the typing rules.
   \item The proof is immediate for $(esk)$,  $(ew0)$, $(ew1)$, $(ecA)$, $(eqA_0)$, $(eqA_1)$, 
     The proof for rules  $(e\beta)$, $(e\TTY)$, $(e\mathrm{if}_l)$, 
     $(e\mathrm{if}_r)$,  $(e\textbf{;})$, $(ecN)$, $(eqN)$  follows
     from the inductive argument.
     The other cases are not possible because of the typing rules.
  \end{enumerate}
\end{proof}

  \vspace{-7ex}

\section{Examples}\label{sec:example}

 \vspace{-1ex}
 
Two further examples of programming in $\QuIA$ follow. One implements 
Deutsch-Jozsa. The other one is (a subroutine of) Simon's algorithm.


\begin{table}[t]
     \centering 
     \begin{tabular}{c @{\qquad\qquad\qquad} c } 
     \vspace{0pt} \includegraphics[scale=0.8]{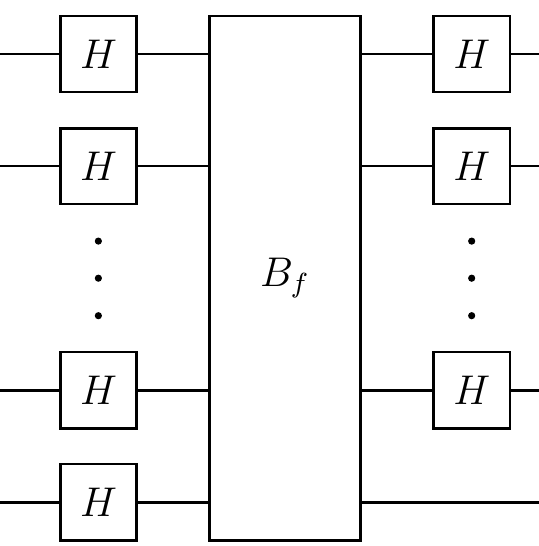}
     &  \vspace{0pt} \includegraphics[scale=0.35]{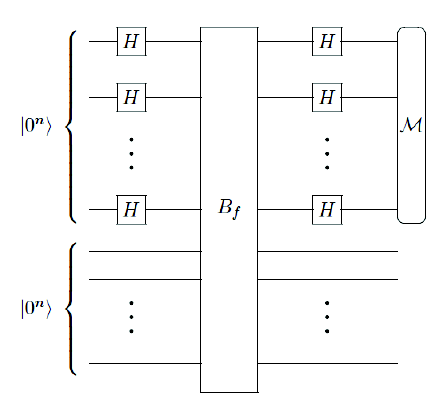}
                                                \\
\end{tabular}
\rule{\textwidth}{1pt}
     \caption{Deutsch-Jozsa circuit (left) and the quantum subroutine of th Simon's algorithm (right).}
     \label{tab:circuits}
   \end{table}

 \begin{example}[Deutsch-Jozsa Circuit]\label{ex:dj}
In this example we show how a \QuIA\ term can represent an infinite family of 
quantum programs which encode Deutsch-Jozsa algorithm~\cite{NieCh10}.
Deutsch-Jozsa is a generalization of Deutsch algorithm that,
given a black-box ${B}_f$ implementing a function $f : \{0,1\} \fleche 
\{0,1\}$, determines whether $f$ is constant or balanced\footnote{
A function is balanced if exactly half of the inputs goes to $0$ and, the other 
half, goes to $1$.} by means of a single call to ${B}_f$, something  impossible 
in the classical case that requires two calls.
\commento{
Deutsch's result plays a major role in quantum computing history, since, 
besides being it the first quantum algorithm ever, provides a separation  
between the complexity classes P and EQP (error-free quantum polytime).  }
Deutsch-Josza solves the parametrized version of the original problem because 
it applies to functions $f : \{0,1\}^{n} \fleche \{0,1\}$.  
The leftmost circuit in Table~\ref{tab:circuits} is a possible 
implementation of Deutsch-Josza in $ \QuIA $. 
When fed with a classical input state $\ket{\underbrace{0\ldots 0}_n1}$, the 
output can be partially measured. Measuring the first $n$ bits tells if the 
function $f$ is constant or not. 
If all $n$ qubits of such a unique measurement are 0, then $ f $ is constant. 
Otherwise, i.e., if at least one of the measurement outcomes is 1, then $ f $ 
is balanced. See~\cite{NieCh10} for further details. 

\medskip
Let $\mathtt{H}^{\;\num{1}}:\CIRC$ be the Hadamard gate and 
$\mathtt{Id}^{\;\num{1}}:\CIRC$ be the Identity gate.
We implement Deutsch-Jozsa in $ \QuIA $ by sequentially composing  $\mathsf{M}_1$,  $\mathtt{x}$ and $\mathsf{M}_3$,
where $\mathtt{x}:\CIRC$ is expected to be substituted by the black-box circuit  that implements the function $f$, while both $\mathsf{M}_1$ and $\mathsf{M}_3$ 
are defined in the coming lines.
\begin{itemize}
\item 
Let $\mpar$ be a term that applied to a circuit ${\tt C}:\CIRC$ and 
to a numeral $\num{n}$ puts $n+1$ copies of ${\tt C}$ in parallel. 
It is defined as
$\mpar=\lambda \mathtt{u}^{\CIRC}.\lambda \mathtt{k}^{\Nat}. Y\mathsf{W_1}\mathtt{u}\mathtt{k}:\CIRC\fleche \Nat\fleche\CIRC,$
where
$\mathsf{W_1}$ is the term $\lambda \mathtt{w}^\sigma.\lambda\mathtt{u}^{\CIRC}.\lambda \mathtt{k}^{\Nat}.  \lif{{k}}{(u)}{\left(  u\;\parallel\; (w\,u\,\Pred (\mathtt{k})) \right)}$
whose type is $\sigma\fleche\sigma$ with $\sigma=\CIRC\fleche\Nat\fleche\CIRC$.

\item 
The circuit $\mathsf{M}_1:\CIRC$ is obtained by feeding the term $\mpar$ with 
two inputs: the (unary) Hadamard gate $\mathtt{H}^{\;\num{1}}$ and the input 
dimension $\rsize{({\tt r})}$ where $r$ is a co-processor register with 
suitable dimension. It should be evident that it generates $n+1$ parallel 
copies of the gate $\mathtt{H}^{\;\num{1}}$.

\item The circuit $\mathsf{M}_3:\CIRC$ can be defined as $(\mpar 
\mathtt{H}^{\num{1}}\;  {\Pred{(\rsize{({\tt r}}))}}) 
\parallel\mathtt{Id}^{\num{1}}   : \CIRC$, i.e. it is obtained by the parallel 
composition of the term $\mpar$ fed by the gate $\mathtt{H}^{\;\num{1}}$ and 
the dimension ${\Pred{(\rsize{({\tt r}}))}}$ (generating $n$ parallel copies of 
the gate $\mathtt{H}^{\;\num{1}}$)  and a single copy $\mathtt{Id}^{\num{1}}$ 
of  the identity gate.
\end{itemize}

Fixed an arbitrary $ n $, the generalization of Deutsch-Jozsa is obtained by using  the quantum variable binder $\QNew {\tt r}{n}{P}$
that makes the quantum variable {\tt r} available in $P$.
In this picture, it is necessary to recall that 
the local variable declaration $\QNew {\tt r}{n}{P}$ creates a quantum 
register which is fully initialized to 0 (Section~\ref{sec:QIA}).
Since the expected input state of Deutsch-Jozsa is 
$\ket{\overbrace{0\ldots0}^{n}1}$, we define and use an initializing circuit
$\mathsf{M_{init}}=(\mpar{\mathtt{Id}^{\;\num{1}}\,}{(\Pred{(\rsize{(r)})})})\parallel \mathtt{Not}^{\;\num{1}}:\CIRC$ 
that complements the last qubit, setting it to $1$.
Let $\mathsf{DJ}^{+}$ be the circuit 
$\mathsf{M_{init}}\,\typecolon\,\mathsf{M}_1\,\typecolon\,\mathtt{x}\,\typecolon\,\mathsf{M}_3$.
The (parametric) $\QuIA$ encoding of Deutsch-Jozsa can be 
defined as
$\lambda \mathtt{x}^\CIRC. \QNew  {\mathtt{r}}  {\;n+1}{(({\tt r}\vartriangleleft\mathsf{DJ}^{+} );\qMeas^{\scalebox{.5}{\num{n}}} {\tt r} )}$. 
The program solves any instance of Deutsch-Jozsa fixed by the 
value of its dimension parameter $ n $ and by providing an encoding of the 
function $f$ to evaluate. 

\medskip

Let $\mathsf{M}_{B_f}$ be a black-box closed circuit  implementing the function $f$ that we want to check
and let  $\mathsf{DJ}^{\star}$ be $\mathsf{DJ}^{+} [ \mathsf{M}_{B_f} /\mathtt{x}]$ namely the circuit obtained by the substitution of $\mathsf{M}_{B_f}$ to $\mathtt{x}$ in $\mathsf{DJ}^{+}$.
By means of the evaluation rule $(EqA_{0})$, we have
$\{({\tt r}, \ket{\underbrace{0\ldots 0}_{n}})\}, {\tt 
r}\vartriangleleft\mathsf{DJ}^{\star} {}\ev_{1} \{({\tt r}, \ket{\phi})\}, 
\none$ where $\ket{\phi}$ is the computational state after the evaluation of 
$\mathsf{DJ}^{\star}$.
To (partially) measure the state $\ket{\phi}$ we use the rule $(eqM)$ to conclude 
$ \{{\tt r}, \ket{\phi}\}, \qMeas^\mathtt{\num{n}} {\tt r}\ev_{1}  \{{\tt r}, 
\ket{\phi'}\}, \num{k}$, where $(k,\ket{\phi'},1)\in 
\partMeas^n(\ket{\phi}, n)$, i.e. $\num{k}$ is the 
(deterministic) output of the measurement and $1$ is the associated probability.
\qed
\end{example}

\medskip

\begin{example}[Simon's algorithm]\label{ex:simon}
  In~\cite{simon94},  Simon exhibited a quantum algorithm that solves in  a polynomial time a problem for which the best known classical algorithm
takes exponential time~\cite{Arora09}.  Simon’s quantum algorithm is an important precursor to Shor’s 
Algorithm for integer factorization (both algorithms are both examples of the Hidden Subgroup Problem over Abelian groups).

We here focus on the the inherently quantum relevant fragment of Simon's algorithm~\cite{KLM07}. 

Simon's problem can be formulated as follows. Let be $f:\{0,1\}^{n}\fleche X$ ($X$ finite) a black-box function.
Determine the string $\mathsf{s}=s_1 s_2 \ldots s_n$  such that $f(x)=f(y)$ if and only if $x=y$ or $x=y\oplus \mathsf{s}$. 
 Simon's algorithm requires an intermediate, partial measure of the quantum  state.
 The measurement is embedded in a quantum subroutine that can be 
 eventually iterated at most $n$ times, where $n$ is the input size. 
 See~\cite{KLM07} for further details and a careful complexity analysis.

The rightmost circuit of Table~\ref{tab:circuits} implements the quantum  subroutine of Simon's algorithm
and has an encoding in $\QuIA$, due to the support of both partial measurement 
and persistent store of quantum measurements.

Simon's subroutine sequentially composes $\mathsf{M}_1$,  $\mathtt{x}$ and $\mathsf{M}_3$,
where $\mathtt{x}:{\CIRC}$ is expected to be substituted by the black-box 
circuit that implements the function $f$ (denoted as $B_{f}$ in 
Table\ref{tab:circuits}). $\mathsf{M}_1$ and $\mathsf{M}_3$ are defined by letting
$\mathsf{M}_1=\mathsf{M}_3=(\mpar (\mathtt{H}^{\;\num{1}}) \rsize{({\tt r})})\parallel (\mpar(\mathtt{Id}^{\;\num{1}})\rsize{({\tt r})}) \;:\CIRC$
where: (i)  $\mpar $ is the term defined in Example~\ref{ex:dj}, 
(ii) ${{\tt r}}$ is a quantum register; and, (iii)  $\mathtt{H}^{\;\num{1}}:\qCom$,
$\mathtt{Id}^{\;\num{1}}:\qCom$
are the unary Hadamard and Identity gates, respectively.

Let $\mathsf{M_{SP}}^{+}$ be the circuit $\mathsf{M}_1\typecolon x \typecolon 
\mathsf{M}_3:\CIRC\quad$. Let $n$ be the arity of $f$ we want to check.
The program that implements Simon's subroutine can be
$\lambda \mathtt{x}^\CIRC. \QNew  {\mathtt{r}}  {\;2*n}{(({\tt r}\vartriangleleft\mathsf{M_{SP}}^{+} );\qMeas^{\scalebox{.5}{\num{n}}} {\tt r} )}$,
where the abstracted variable $x:\CIRC$ will be replaced by a suitable encoding 
of the black-box function that implements $f$.


Let $\mathsf{M}_{B_f}\;:\CIRC$ be  the encoding of the circuit implementing $f$ 
and 
let $\mathsf{M_{SP}}^{\star}$ be $\mathsf{M_{SP}}^{+} [ \mathsf{M}_{B_f} 
/\mathtt{x}]$, namely the circuit obtained by the substitution of 
$\mathsf{M}_{B_f}$ for $\mathtt{x}$ in $\mathsf{M_{SP}}^{+}$.

It is easy to check that the following evaluation respects the $\QuIA$ semantics (rule $(EqA_{0})$):
$$\{({\tt r}, \ket{\underbrace{0\ldots 0}_{2*n}})\}, {\tt r}\vartriangleleft {\mathsf{M_{SP}}^{\star}}\ev_{1} \{({\tt r}, \ket{\phi})\}, \none \quad,$$ where $\ket{\phi}$ is the state after the evaluation of the circuit $\mathsf{M_{SP}}^{\star}$.
We can measure the first $\num{n}$ quantum bits as follows:
$ \{({\tt r}, \ket{\phi})\}, \qMeas^\mathtt{\num{n}} {\tt r}\ev_{\alpha}  
\{({\tt r}, \ket{\phi'})\}, \num{k}$, where $(k,\ket{\phi'},\alpha)\in 
\partMeas^{2*n}(\vettore{s}'(\mathtt{r}), n)$.

The classical output $\num{k}$ can be used as a feedback from the quantum co-processor by the classical program, 
in this way it can decide how to proceed in the computation. In particular, it can use the measurement as guard-condition in a loop 
that iterates the subroutine. 
So we can easily re-use the Simon-circuits above as many times as we want, by 
arbitrarily reducing the probability error.
\qed
\end{example}

\vspace{-1ex}
 \vspace{-2ex}

\section{Conclusions and future work}\label{sec:conclusions}
$\QuIA$ is a higher-order programming language that manages quantum 
co-processors. We formalize co-processors as quantum registers that store 
quantum states. This approach is radically new w.r.t. the existing proposals 
due to the following distinctive features:
(i) each quantum variable is associated to a unique quantum state,
we can duplicate such a name at will without invalidate the linear constraints that the quantum state has to satisfy; 
(ii) we formalize an elegant hiding mechanism that provides a natural approach to multiple co-processors internalized in the language;
and, (iii) the classical programming constructs included in $ \QuIA $ can be used naturally by a traditional programmer,
because they are unaffected by the generally quite restrictive requirements about the management of quantum data.
This approach introduces a neat separation between the description of the directives 
to manipulate states in quantum co-processors from the names for quantum states.
The reason is that directives 
are circuits that we consider as classical data that are freely duplicable and erasable.
Since the wide expressiveness of Idealized Algol is preserved in $\QuIA$, we think we are proposing a programming tool which 
represent a further step in the design of quantum programming languages, 
coherently with directions that Knill advocates in~\cite{Knill96}.


Current ongoing work focuses on semantics and typing systems of $ \QuIA $.
First, we plan to add dependent types for circuits and registers, in analogy 
to~\cite{paolini2017lncs,qwire}.
Second, we are studying a mature approach to the representation of quantum circuits by making explicit the linear management of their wires
(namely, a revised and restricted version of the circuits considered in~\cite{qwire}).
Third, we are interested in the formalization of a call-by-value version of 
$\QuIA$. The goal is to further ease the embedding 
of quantum programming in traditional programming frameworks.
Fourth, we are interested in developing a denotational semantics for $\QuIA$, 
maybe a not complete one, but suitable to tackle the equivalence between 
programs involving (meaningful) quantum, non-deterministic~\cite{AschieriZ13,AschieriZ16}, probabilistic and reversible~\cite{PaoliniPiccoloRoversi-ENTCS2016,paolini2018ngc}
aspects~\cite{DLZ12}.


\vspace{-3ex}

\bibliographystyle{eptcs}
\bibliography{biblio}

\end{document}